\documentclass{article}

\usepackage{url}
\usepackage{amssymb}
\usepackage{amsmath}
\usepackage{amsthm}

\newtheorem{lemma}{Lemma}[section]
\newtheorem{proposition}[lemma]{Proposition}
\newtheorem{theorem}[lemma]{Theorem}
\newtheorem{corollary}[lemma]{Corollary}
\theoremstyle{definition}
\newtheorem{example}[lemma]{Example}
\newtheorem{definition}[lemma]{Definition}
\newtheorem{remark}[lemma]{Remark}

\newcommand{\ar}[1]{\mathrm{ar}(#1)}
\newcommand{\var}[1]{\mathrm{var}(#1)}
\newcommand{\oids}[1]{{\mathrm{oids}(#1)}}
\newcommand{\consts}[1]{{\mathrm{consts}(#1)}}
\newcommand{\sotgd}[1]{\mathrm{sotgd}(#1)}
\newcommand{\ntgd}[1]{\mathrm{ntgd}(#1)}
\newcommand{\adom}[1]{\mathrm{adom}(#1)}

\newcommand{\dom}{\mathbf{dom}}

\newcommand{\id}{\mathrm{id}}

\newcommand{\tpl}[1]{\bar{#1}}		
\newcommand{\set}[1]{\{\,#1\,\}}        
\newcommand{\col}{\colon}

\newcommand{\schema}[1]{{\textbf{#1}}}

\newcommand{\scenario}[1]{{\cal #1}}

\newcommand\M{\scenario M}

\newcommand{\fltg}[1]{\mathring{#1}}

\newcommand{\Mat}[2]{\mathrm{Mat}(#1,#2)}

\renewcommand{\setminus}-

\begin{document}


\title{Mapping-equivalence and oid-equivalence of single-function
object-creating conjunctive queries}

\author{Angela Bonifati         \and
        Werner Nutt \and
	Riccardo Torlone\and
        Jan Van den Bussche 
}

\maketitle

\begin{abstract}

Conjunctive database queries have been extended with a mechanism
for object creation to capture important applications such as
data exchange, data integration, and ontology-based data access.
Object creation generates new object identifiers in the result,
that do not belong to the set of constants in the source
database. The new object identifiers can be also seen as Skolem
terms.  Hence, object-creating conjunctive que\-ries can also be
regarded as restricted second-order tuple-generating dependencies
(SO tgds), considered in the data exchange literature.

In this paper, we focus on the class of single-function
object-creating conjunctive queries, or sifo CQs for
short.  The single function symbol can be used only once in the
head of the query.  We
give a new characterization for oid-equiv\-a\-lence of sifo CQs that
is simpler than the one given by Hull and Yoshikawa and places
the problem in the complexity class NP\@.  Our characterization
is based on Cohen's equivalence notions for conjunctive queries
with multiplicities.  We also solve the logical entailment
problem for sifo CQs, showing that also this problem belongs to
NP\@.  Results by Pichler et al.\ have shown that logical
equivalence for more general classes of SO tgds is either
undecidable or decidable with as yet unknown complexity upper
bounds.

\end{abstract}

\section{Introduction}
\label{sec:intro}

Conjunctive queries form a natural class of database queries,
which can be defined by combinations of selection,
renaming, natural join, and projection.  Much of the research on
database query processing is focused on conjunctive queries;
moreover, these queries are ame\-na\-ble to advanced optimizations
because containment of conjunctive queries is decidable (though
NP-complete).  In this paper, we are interested in conjunctive
queries extended with a facility for object creation.

Object creation, also called oid generation or value invention,
has been repeatedly proposed and investigated as a feature of
query languages.  This has happened in several contexts: high
expressiveness \cite{av_proc,av_datalog,vvag_compl}; object
orientation
\cite{ak_iql,vdbp_abstr,hy_ilog,kw_pods89,maier_ologic}; data
integration \cite{tsimmis}; semi-structured data and XML
\cite{abs_book}; and data exchange
\cite{miller_invention,clio_overview,sotgd}.  In a logic-based
approach, object creation is typically achieved through the use
of Skolem functions \cite{hy_ilog,kw_pods89,maier_ologic}.

In the present paper, we consider conjunctive queries (CQs)
extended with object creation through the use of a single Skolem
function, which can be used only
once in the head of the query. We refer to such a query as a `sifo CQ' (for
\underline{si}ngle-\underline function \underline
object-creating). 
The following example of a sifo CQ uses a Skolem
function~$f$:
$$
Q:\mathit{Family}(c,f(x,y)) \gets \mathit{Mother}(c,x),
\mathit{Father}(c,y).
$$
The query introduces a new oid $f(x,y)$ for
every pair $(x,y)$ of a woman $x$ and a man $y$ who have at least
one child together; all children $c$ of $x$ and $y$ are linked to the
new oid in the result of the query (a relation called
$\mathit{Family}$).
As an example, if $\mathit{Mother(beth,anne)}$ and $\mathit{Father(beth,adam)}$ are
two facts in the underlying database, then the result of the query includes the fact
$\mathit{Family(beth,f(anne,adam))}$, where $\mathit{f(anne,adam)}$ is the
newly created oid. This oid will be shared by all the children having $\mathit{anne}$ and $\mathit{adam}$ as parents.

In this paper, we first revisit the problem of checking oid-equivalence
of sifo CQs. Oid-equivalence has its origins in the
theory of object-creating queries introduced by
Abiteboul and Kanellakis~\cite{ak_iql}; it is the natural generalization
of query equivalence in the presence of object
creation.

Consider for instance the following sifo CQ:
$$Q':\mathit{Family}(c,g(x,y,x)) \gets \mathit{Mother}(c,x),
\mathit{Father}(c,y).$$
%
%
It is not hard to see that the result of $Q'$ has the same structure
as the result of the query $Q$ above. 
The query $Q'$ links all children $c$ of the parents $x$ and $y$ to the oid 
$g(x,y,x)$ that depends exactly on $x$ and $y$. 
That is, two children in the result of $Q$ are connected to the same oid if
and only if they are connected to same oid in $Q'$, although the oids
will be syntactically different.
Therefore, we can conclude that $Q$ and $Q'$ are oid-equivalent,
which means that their results are identical on any input up to a
simple isomorphism mapping the oids in one result to those in the other.

Hull and Yoshikawa \cite{hy_pods91} studied oid-equivalence (they
called it `obscured equivalence') for nonrecursive ILOG programs; the
decidability of this problem is a long-standing open question.
Nevertheless, for the case of `isolated oid creation', to which sifo
CQs belong, they have given a decidable characterization.

We give a new result relating oid-equivalence to equivalence
of classical conjunctive queries under `combined' bag-set semantics
\cite{cohen_multiplicities}, which models the evaluation of CQs when query results and
relations may contain duplicates of tuples.  As a corollary, we obtain that
oid-equivalence for sifo CQs belongs to NP, which does not follow
from the Hull-Yoshikawa test.
Obviously, then, oid-equiv\-a\-lence for sifo CQs is
NP-complete, since equivalence of classical CQs without
object creation is already NP-complete.

Object creation is receiving renewed interest in the
context of schema mappings~\cite{miller_invention,sotgd}, which are formalisms describing how data structured under a source schema are to be transformed into data structured under a target schema.  Hence, it is
instructive to view sifo CQs as schema mappings, simply by
interpreting them as implicational statements.
As an example, we may view query $Q$ above as an implicational
statement that relates a
query over relations $\mathit{Mother}$ and $\mathit{Father}$ in the source schema
to the relation $\mathit{Family}$ in the target schema.

For standard CQs
without object creation, two que\-ries are equivalent if and only
if they are logically equivalent as schema mappings
\cite{nash_optimization}.
For sifo CQs, we show that
oid-equivalence implies logical equivalence, while the converse is not true.

Sifo CQs viewed as schema mappings belong to the class of
so-called `nested dependencies' \cite{miller_invention}, which belong in
turn to the class of formulas called second-order tuple-generating dependencies
(SO-tgds \cite{sotgd}).
For instance, consider again the sifo CQ $Q$ above: it can be rewritten into the following 
SO-tgd:
$$
\exists f \forall x \forall y \forall c (\mathit{Mother}(c, x)
\land \mathit{Father} (c,y)
{} \to \mathit{Family}(c,f(x,y))),
$$
which is of second order because the function $f$ is existentially quantified.

Although logical equivalence of SO-tgds is undecidable
\cite{pichler_sotgd}, logical implication of nested
dependencies has recently been shown to be decidable
\cite{pichler_nestedep}.  We give a novel and elegant
characterization of logical
implication for sifo CQs which is simpler than the general
implication test for nested dependencies.  It turns out that
the problem belongs to NP\@.  Hence, logical implication for
sifo CQs has no worse complexity than containment for standard
CQs without object creation.

Summarizing, in this paper we provide the following contributions in the
area of query languages with object creation:
\begin{enumerate}
\item
We clarify the relationship between sifo CQs and other formalisms
in the literature, notably, the language ILOG \cite{hy_ilog}, second-order
tuple-generating dependencies \cite{sotgd}, and nested tuple-generating
dependencies \cite{miller_invention}.
\item
We relate the problem of oid-equivalence
for sifo CQs to the equivalence of classical conjunctive
queries under combined bag-set semantics, which implies its NP-completeness.
\item
We show that when sifo CQs are interpreted as schema mappings, oid-equivalence
implies logical equivalence but not vice versa.
\item
We provide a new characterization of logical implication for sifo CQs as object-creating queries showing  that this problem has the same complexity as deciding containment for classical CQs.
\end{enumerate}

This paper is organized as follows.
In Section~2 we review some practical applications of sifo CQs.
In Section~3 we formally define object-creating conjunctive
queries.  Section~4 is devoted to the results on oid equivalence.
Section~5 is devoted to the results on logical entailment.
In Section~6 we conclude by discussing
related work and topics for further research.

\section{Applications of sifo CQs}

In this section, we discuss further applications of sifo CQs, which may constitute important components
of many advanced database systems, spanning from information integration
and schema mapping engines along with their benchmarks,
to several Semantic Web tools.  We believe this shows that the
results in this article on equivalence and logical implication of
sifo CQs are relevant and contribute
to our understanding of how solutions for these applications can
be optimized.

GAV ({\em global-as-view}) schema mappings
\cite{FriedmanLM99,Lenzerini02,Ullman00} relate a query over the source
schema, represented by a body $B$ of a CQ, to an atomic element of the global
schema, represented by a head atom $H$
of a CQ\@. More precisely, a GAV mapping can be written as follows:
$$
T(\tpl x) \gets B
$$
where we use a relation symbol $T$ as the atomic head predicate.

GAV schema mappings have been used already in the 1990s in mediator
systems like Tsimmis \cite{PapakonstantinouGW95,Ullman00} or Information Manifold \cite{LevyRO96} for the
integration of heterogeneous data sources. In both systems, source facts are related
to facts over the global schema by means of queries.

Sifo CQs can naturally be seen as extensions of GAV mappings,
when one of the attributes of the global schema carries newly created identifiers.
For instance, the sifo CQ $Q$ from Section \ref{sec:intro} can
express a mapping from a source schema containing two relations
$\mathit{Mother}$ and $\mathit{Father}$ to one relation
$\mathit{Family}$ of a global schema, with created identifiers
for families appearing in the tuples in the result of the mapping. Thus, we
can also interpret $Q$ as an extended GAV schema mapping.

Another important application of sifo CQs are schema mapping
benchmarks allowing the users to compare and evaluate schema mapping
systems.  In particular, the flexibility of the
arguments of the Skolem functions used for object creation
has been advocated as one of the desirable features in recent benchmarks for schema mapping and information integration, such as STBenchmark
\cite{AlexeTV08}  and iBench \cite{ACGM15}.


More precisely, in the mapping primitives of iBench \cite{ACGM15}, an
extension of STBenchmark \cite{AlexeTV08} that supports SO-tgds, the
users can choose among two different skolemization strategies to fill
the arguments of the Skolem functions: 
\emph{fixed}, where the arguments of the function are pre-defined in a native 
mapping primitive,
or \emph{variable},
where one can further choose among the options 
\emph{All, Key}, and \emph{Random}, 
which generate mappings where 
all variables, 
the variables in the positions of the primary key, or
a random set of variables, respectively, 
are used as arguments of the function.

%
These skolemization strategies can be captured by sifo CQs as follows. 
In the query below:
$$
T(x, y, f(x, y, z, w)) \gets B (x, y, z, w)
$$
we can observe that the Skolem term uses all the source variables
in the body $B$ (option {\em All\/}).
If the attribute in the position of $x$ is a primary key for~$B$,
then the application of the option {\em Key\/} generates a mapping
that can be expressed by the sifo CQ 
$$
T(x, y, f(x)) \gets B (x, y, z, w).
$$
%

Alternatively, choosing the option {\em Random} may lead the iBench to
randomly select the attributes in the positions of $x$ and $z$, 
and then to generate the mapping represented by
$$
T(x, y, f(x, z)) \gets B (x, y, z, w).
$$

It is also worth highlighting that three out of the seven
mapping primitives in iBench that are novel with respect to
STBenchmark, namely {\em ADD} (copy a relation and ADD new
attributes), {\em ADL} (copy a relation, Add and DeLete
attributes in tandem) and {\em MA} (Merge and Add new attributes)
contain single Skolem functions. They correspond to the following sifo
CQs, respectively:
\begin{align*}
& T(x, y, f(x, y)) \gets B (x, y) \\
& T(x, f(x)) \gets B (x, y) \\
& T(x, y, z, f(x, y, z)) \gets B (x, y), T(y, z).
\end{align*}

A third significant application of sifo CQs is the Semantic Web, where sifo
CQs can be envisioned in at least two scenarios, namely in systems for
ontology-based data access (OBDA) and in direct mappings from the relational to
the RDF data format, under development at W3C.%
\footnote{\url{http://www.w3.org/TR/rdb-direct-mapping/}}
Indeed, newly created identifiers in the head of a sifo CQ can
serve as generated keys, or simply as newly invented values
needed to fill an attribute of a relation in the global schema.
As such, sifo CQs can be seen as examples of mapping assertions from source schemas to a
global ontology in OBDA \cite{PoggiLCGLR08}.
Typically, OBDA mapping assertions relate facts in relational source schemas to RDF
triples in a global ontology. The newly generated IRIs%
\footnote{IRIs stand for Internationalized Resource Identifiers and
  extend the syntax of URIs (Uniform Resource Identifiers) to a much
  wider repertoire of characters.  They naturally embody global
  identifiers that refer to the same resource on the Web and can be
  used across different mapping assertions to refer to that resource.}
in the RDF triples can be interpreted as skolemized values in the
global ontology.

A related application is the
direct translation of a
relational schema into OWL, which uses as an important building block the
creation of IRIs \cite{SequedaArenasMarcelo-WWW2012}.
In contrast to the previous application, this application handles
relational schemas that are not known in advance.
For each relation $r$ in a database schema,
Datalog-like rules can be used to generate an IRI
for the relation $r$ and an IRI for each attribute $a$ in $r$.
We take an example of a translation from a relational schema into OWL
and we show that, actually, these Datalog-like rules can be viewed as
sifo CQs, since they employ a single concatenation function to obtain
such IRIs (exemplified as $f$).
The corresponding sifo CQs are reported below:
\begin{align*}
& T_1(r, f(b, r)) \gets B_1(r ) \\
& T_2(a, r, f(b, r, a)) \gets B_2 (r, a),
\end{align*}
where $B_1$ and $B_2$ are conjunctive query
bodies retrieving relation names $r$ and
attribute names $a$ from the data dictionary of an underlying relational database, and where $b$ is a
string representing a given \emph{IRI base} (e.g., the string
`http://example.edu/db') for the same database to
be translated. Thus, the first query creates a new IRI for the
relation $r$, by concatenating $b$ with the relation symbol $r$, while
the second query returns the set of IRIs of the attributes
$a$ of $r$, by concatenating $b$ with the relation symbol $r$ and
its attribute symbols~$a$.




\section{Preliminaries} \label{prelim}

In this section we introduce our formalism for dealing with
conjunctive queries and introduce the notion of object-creating
conjunctive query, adapted from the language ILOG
\cite{hy_ilog}.

\subsection{Databases and conjunctive queries} \label{secdbcq}

From the outset we assume a supply of \emph{relation
names}, where each relation name $R$ has an associated arity
$\ar R$.  We also assume an infinite domain $\dom$ of
atom\-ic data elements called \emph{constants}.  A
\emph{fact} is of the form $R(a_1,\dots,a_k)$ where
$a_1$, \dots, $a_k$ are constants
and $R$ is a $k$-ary relation name.
We call $R$ the \emph{predicate} of the fact.

A \emph{database schema} $\schema S$ is a finite set of relation
names.  An \emph{instance} of $\schema S$ is a finite set of
facts with predicates from $\schema S$.  The set of all constants
appearing in an instance $I$ is called the \emph{active domain}
of $I$ and denoted by $\adom I$.

We further assume an infinite supply of \emph{variables},
disjoint from $\dom$.  An \emph{atom} is of the form
$R(x_1,\dots,x_k)$ where $x_1,\dots,x_k$ are variables and
$R$ is a $k$-ary relation name.  As with facts, we call $R$
the predicate of the atom.

We can now recall the classical notion of conjunctive query (CQ)
\cite{ahv_book,cm}.  Syntactically, a CQ over a database schema
$\schema S$ is of the form 
$$
  H \gets B,
$$
where $B$ is a finite
set of atoms with predicates from $\schema S$, and $H$ is an atom
with a predicate not in $\schema S$.  The set $B$ is called the
\emph{body} and $H$ is called the \emph{head}.  It is required
that every variable occurring in the head also occurs in the
body.  We denote the set of variables occurring in a set of atoms
$B$ (or a single atom $A$) by $\var B$ (or $\var A$).

The semantics of CQs is defined in terms of valuations.  A
\emph{valuation} is a mapping $\alpha : X \to \dom$ on some finite
set of variables $X$.  When $A$ is an atom with $\var A \subseteq
X$, we can apply $\alpha$ to $A$ simply by applying $\alpha$ to
every variable in $A$.  This results in a fact and is denoted by
$\alpha(A)$.  When $B$ is a set of atoms and $\alpha$ is a
valuation on $\var B$, we can apply $\alpha$ to $B$ by
applying $\alpha$ to every atom in $B$. Formally, $\alpha(B)$ is
defined as the instance $\{\alpha(A) \mid A \in B\}$.

When $I$ is an instance and $\alpha$ is a valuation on $\var B$
such that $\alpha(B) \subseteq I$,
we say that $\alpha$ is a \emph{matching} of $B$ in
$I$, and denote this by $\alpha : B \to I$.  Now when $Q$ is a
CQ $H \gets B$ and $I$ is an instance, the result of $Q$ on $I$
is defined as $$ Q(I) := \{\alpha(H) \mid \alpha : B \to I\}. $$

\subsection{Object-creating conjunctive queries}

Assume a finite vocabulary of \emph{function symbols} of various
arities. As with relation names, the arity of a function symbol $f$
is denoted by $\ar f$.

\emph{Data terms} are syntactical
expressions built up from
constants using function symbols.  Formally,
data terms are inductively defined as follows:
\begin{enumerate}
\item
Every constant is a data term;
\item
If $f$ is a $k$-ary function symbol and
$d_1,\dots,d_k$ are data terms, then the expression
$f(d_1,\dots,d_k)$ is also a data term.\footnote{Since constants are
atomic data elements, no constant is allowed to be of the form
$f(d_1,\dots,d_k)$.}
\end{enumerate}

An \emph{extended fact} is defined just like a fact, except that
it may contain data terms rather than only constants.  Formally,
an extended fact is of the form $R(d_1,\allowbreak
\dots,\allowbreak d_k)$, 
where $d_1,\dots,d_k$ are data terms and $R$ is a $k$-ary relation name.
The active domain of an extended fact $e =
R(d_1,\dots,d_k)$
is defined as $$ \adom e := \{d_1,\dots,d_k\}. $$
An \emph{extended instance} is a finite set of extended facts.
The active domain of an extended instance $J$ is defined as
$$ \adom J := \bigcup_{e \in J} \adom e. $$

\emph{Formula terms} are defined in the same way as data terms,
but are built up from variables rather than constants.
\emph{Extended atoms} are defined like atoms, but can
contain formula terms in addition to variables.
If $t$ is a formula term and $\alpha$ is a valuation
defined on all variables occurring in $t$, we can apply $\alpha$
to every variable occurrence in $t$, obtaining a data term
$\alpha(t)$.  Likewise, we can apply a valuation to an extended
atom, resulting in an extended fact.

We are now ready to define the syntax and semantics of
\emph{object-creating conjunctive queries (oCQ).}  Like a
classical CQ, an oCQ is of the form $H \gets B$.  The only
difference with a classical CQ is that $H$ can be an extended
atom; in particular, $B$ is still a finite set of ``flat'' atoms,
not extended atoms.  It is still required that $\var H \subseteq
\var B$.  The result of an oCQ $Q = H \gets B$ on an
instance $I$ is now an extended instance, defined as $$ Q(I) :=
\{\alpha(H) \mid \alpha : B \to I\}. $$

\begin{table}
\centering
\begin{tabular}[t]{|ll|}
\multicolumn{2}{c}{Mother} \\
\hline
beth & anne \\
ben & anne \\
eric & claire \\
emma & diana \\
dave & diana \\
\hline
\end{tabular}
\qquad
\begin{tabular}[t]{|ll|}
\multicolumn{2}{c}{Father} \\
\hline
beth & adam \\
ben & adam \\
eric & carl \\
emma & carl \\
\hline
\end{tabular}
\qquad
\begin{tabular}[t]{|ll|}
\multicolumn{2}{c}{Family} \\
\hline
beth & $f(\rm anne,adam)$ \\
ben & $f(\rm anne,adam)$ \\
eric & $f(\rm claire,carl)$ \\
emma & $f(\rm diana,carl)$ \\
\hline
\end{tabular}
\caption{Instances used in Example~\ref{exfamily}.}
\label{tabfamily}
\end{table}

\begin{table}
$$
\begin{array}[t]{|lll|}
\multicolumn{3}{c}{R} \\
\hline
a & b & c \\
a & b & d \\
c & b & d \\
d & c & a \\
\hline
\end{array}
\qquad
\begin{array}[t]{|ll|}
\multicolumn{2}{c}{T} \\
\hline
a & f(b) \\
c & f(b) \\
d & f(c) \\
\hline
\end{array}
$$
\caption{Instances used in Example~\ref{exar}.}
\label{tabar}
\end{table}

\begin{example} \label{exfamily}
Recall the oCQ $Q$ from the Introduction:
$$
\mathit{Family}(c,f(x,y)) \gets \mathit{Mother}(c,x),
\mathit{Father}(c,y). 
$$
If $I$ is the instance consisting of the Mother and Father facts
listed in Table~\ref{tabfamily}, then $Q(I)$ is the extended
instance consisting of the
extended Family facts listed in the same table.
\end{example}

\begin{example} \label{exar}
For a more abstract example, consider the following oCQ $Q$:
$$ T(x,f(y)) \gets R(x,y,z). 
$$
If $I$ is the instance consisting of the $R$-facts listed in
Table~\ref{tabar}, then $Q(I)$ consists of the extended $T$-facts
listed in the same table.
\end{example}

\subsection{The single-function case}

In this paper, we focus on \emph{single-function} oCQs (sifo
CQs), that have exactly one occurrence of a function symbol in the
head.  Without loss of generality we always place the function
term in the last position of the head.

\begin{definition}
A sifo CQ over a database schema $\schema S$ is an oCQ over
$\schema S$ of the form
$$ 
  T(\bar x,f(\bar z)) \gets B, 
$$
where
\begin{itemize}
\item
$T$ is the head predicate;
\item
$f$ is a function symbol;
\item
$B$ is the body;
\item
$\bar x$ is a tuple of (not necessarily distinct) variables from
$\var B$, called the \emph{distinguished variables};
\item
$\bar z$ is a tuple of (not necessarily distinct) variables from
$\var B$, called the \emph{creation variables};
some creation variables may be distinguished;
\item
The elements of $\var B$ that are not distinguished are called
the \emph{non-distinguished variables}.
\end{itemize}
\end{definition}


\begin{example} \label{exsifo}
The queries in Examples
\ref{exfamily} and \ref{exar} are both examples of sifo CQs.
\end{example}

\subsection{Comparison with ILOG}

Object-creating CQs can be considered to be the con\-junc\-tive-query
fragment of nonrecursive ILOG
\cite{hy_ilog}; our syntax exposes the Skolem functions, which are
normally obscured in the standard ILOG syntax, and our semantics
corresponds to what is called the `exposed semantics' by Hull and
Yoshikawa.  Nevertheless, in the following section, we will
consider oid-equivalence of sifo CQs, which does correspond to what
has been called `obscured equivalence' \cite{hy_pods91}.

\section{Characterization of oid-equivalence for sifo CQs}
\label{secoidequiv}


\subsection{Oid-equivalence of oCQs}

The result $Q(I)$ of an oCQ $Q$ applied to an instance $I$ is an
extended instance.  The data terms in $\adom{Q(I)}$ that are not
constants play the role of created oids (also called invented values).
Intuitively it is clear that the actual form of the created oids
does not matter.

\begin{table}
\centering
\begin{tabular}[t]{|ll|}
\multicolumn{2}{c}{Family} \\
\hline
beth & $g(\rm anne,adam,anne)$ \\
ben & $g(\rm anne,adam,anne)$ \\
eric & $g(\rm claire,carl,claire)$ \\
emma & $g(\rm diane,carl,diane)$ \\
\hline
\end{tabular}
\caption{Instance used in Example~\ref{exequiv}.}
\label{tabequiv}
\end{table}

\begin{example} \label{exequiv}
Recall the query $Q$ from Example~\ref{exfamily}:
$$
\mathit{Family}(c,f(x,y)) \gets \mathit{Mother}(c,x),
\mathit{Father}(c,y). $$
As mentioned in the Introduction, 
we could have used equivalently the following query $Q'$:
$$
\mathit{Family}(c,g(x,y,x)) \gets \mathit{Mother}(c,x),
\mathit{Father}(c,y). $$
Applying the above query to the Mother and Father facts from
Table~\ref{tabfamily} results in the instance shown in
Table~\ref{tabequiv}.  Intuitively, this instance has exactly the
same relevant properties as the Family-instance from
Table~\ref{tabfamily}: beth and ben are linked to the same
family-oid; eric is linked to another oid;
and emma to still another one.
\qed
\end{example}

We formalize this intuition in the following definitions.

\begin{definition}
Let $J$ be an extended instance.
\begin{itemize}
\item
The set $\adom J - \dom$ is denoted by $\oids J$;
\item
The set $\adom J \cap \dom$ is denoted by $\consts J$.
\end{itemize}
\end{definition}

\begin{definition}
Let $J$ be an extended instance and let $\rho$ be a mapping from
$\adom J$ to the set of data terms.  For any extended fact $e =
R(d_1,\dots,d_k)$ in
$J$, we define $\rho(e)$ to be the extended fact
$R(\rho(d_1),\dots,\rho(d_k))$.
We then define $\rho(J) := \{\rho(e) \mid e \in J\}$.
\end{definition}

\begin{definition}
Let $J_1$ and $J_2$ be extended instances.  Then $J_1$ and $J_2$
are called \emph{oid-isomorphic} if there exists a bijection
$\rho : \adom {J_1} \to \adom {J_2}$ such that
\begin{itemize}
\item
$\rho$ is the identity on
$\consts {J_1}$;
\item
$\rho$ maps $\oids {J_1}$ to $\oids {J_2}$;
\item
$\rho(J_1)=J_2$.
\end{itemize}
Such a bijection $\rho$ is called an \emph{oid-isomorphism} from
$J_1$ to $J_2$.
\end{definition}

The above definition implies that oid-isomorphic instances have
the same constants.  Formally, if $J_1$ and $J_2$ are oid-isomorphic
then $\consts {J_1} = \consts {J_2}$.

\begin{definition}
Let $Q$ and $Q'$ be two oCQs with the same head predicate, and
over the same database sche\-ma~$\schema S$.
Then $Q$ and $Q'$ are called \emph{oid-equivalent} if for
every instance $I$ over $\schema S$, the results $Q(I)$ and
$Q'(I)$ are oid-isomorphic.
\end{definition}

\begin{example} \label{exoideq}
The queries in Example \ref{exequiv} are oid-equivalent.
For example, for the instance $I$ of Table~\ref{tabfamily}, the
oid-isomorphism from $Q(I)$ to $Q'(I)$ is as follows:
$$ \begin{array}{rcl}
f({\rm anne,adam}) & \mapsto & g(\rm anne,adam,anne) \\
f({\rm claire,carl}) & \mapsto & g(\rm claire,carl,claire) \\
f({\rm diane,carl}) & \mapsto & g(\rm diane,carl,diane).
\end{array}
$$
\end{example}

\begin{table}
$$
\begin{array}[t]{|lll|}
\multicolumn{3}{c}{I} \\
\hline
a & b & c \\
d & b & e \\
\hline
\end{array}
\qquad
\begin{array}[t]{|ll|}
\multicolumn{2}{c}{Q(I)} \\
\hline
a & f(b) \\
d & f(b) \\
\hline
\end{array}
\qquad
\begin{array}[t]{|ll|}
\multicolumn{2}{c}{Q'(I)} \\
\hline
a & f(a,b) \\
d & f(d,b) \\
\hline
\end{array}
$$
\caption{Instances used in Example~\ref{exnotoidequiv}.}
\label{tabnotoidequiv}
\end{table}

\begin{example}
\label{exnotoidequiv}
Recall the query $Q$ from Example~\ref{exar}:
$$ T(x,f(y)) \gets R(x,y,z) $$
Also consider the following variation $Q'$ of $Q$:
$$ T(x,f(x,y)) \gets R(x,y,z) $$
Then $Q$ and $Q'$ are not oid-equivalent, as shown by the
simple instances in Table~\ref{tabnotoidequiv}.  Indeed, there cannot be
an oid-isomorphism from $Q(I)$ to $Q'(I)$ because $Q(I)$ contains
only one distinct oid while $Q'(I)$ contains two distinct oids.
\end{example}

\begin{table}
$$
\begin{array}[t]{|lll|}
\multicolumn{3}{c}{I} \\
\hline
a & b & c \\
a & d & e \\
\hline
\end{array}
\qquad
\begin{array}[t]{|ll|}
\multicolumn{2}{c}{Q(I)} \\
\hline
a & f(a) \\
\hline
\end{array}
\qquad
\begin{array}[t]{|ll|}
\multicolumn{2}{c}{Q'(I)} \\
\hline
a & f(a,b,c) \\
a & f(a,d,e) \\
\hline
\end{array}
$$
\caption{Instances used in Example~\ref{ex2notoid}.}
\label{tab2notoid}
\end{table}

\begin{example}
\label{ex2notoid}
As a variant of Example~\ref{exnotoidequiv}, consider the following two
oCQs:
\begin{align*}
Q & = T(x,f(x)) \gets R(x,y,z) \\
Q' & = T(x,f(x,y,z)) \gets R(x,y,z)
\end{align*}
Again these two oCQs are not oid-equivalent, as shown by the
counterexample instances in
Table~\ref{tab2notoid}.
\end{example}

\subsection{Homomorphisms and containment of conjunctive queries}

The characterizations we will give for oid-equivalence of sifo
CQs depend on the classical notions of homomorphism and
containment between conjunctive queries.
Let us briefly recall these
notions now \cite{cm,ahv_book}.

A \emph{variable mapping} is a mapping $h$ from a finite set $X$
of variables to another finite set $Y$ of variables.  If $A$ is
an atom with variables in~$X$, then we can apply $h$ to each
variable occurrence in~$A$ to obtain an atom with
variables in $Y$, which we denote by $h(A)$.  If $B$ is a set of
atoms with $\var B \subseteq X$,
then we naturally define $h(B):=\{h(A) \mid A \in B\}$.

For two sets $B$ and $B'$ of atoms, a variable mapping $h : \var
B \to \var {B'}$ is called a \emph{homomorphism from $B$ to $B'$}
if $h(B) \subseteq B'$.  This is denoted by $h : B \to B'$.  The
notion of homomorphism is extended to conjunctive queries $Q = H
\gets B$ and $Q' = H' \gets B'$ as follows.  A homomorphism from
$Q$ to $Q'$ is a homomorphism $h : B \to B'$ such that $h(H)=H'$.
This is denoted by $h : Q \to Q'$.

A classical result relates
homomorphisms between conjunctive queries to containment.  Let
$Q$ and $Q'$ be two conjunctive queries over a common database schema
$\schema S$. We say that \emph{$Q'$ is contained in $Q$} if for
every instance $I$ of $\schema S$, we have $Q'(I) \subseteq
Q(I)$.  The classical result states that $Q'$ is contained in $Q$
if and only if there exists a homomorphism $h : Q \to Q'$.

Two queries $Q$ and $Q'$ are \emph{equivalent} if for every
instance $I$ of $\schema S$, we have $Q(I)=Q'(I)$.  Since
equivalence amounts to containment in both directions,
two conjunctive queries are equivalent if and only if there
exist homomorphisms between them in both directions.

\subsection{A normal form for oid-equivalence problems}

In this subsection we
consider two arbitrary sifo CQs $Q$, $Q'$ with the same head predicate:
\begin{align*}
Q \ &= \ T(\tpl x,f(\tpl z)) \gets B
\\
Q'\ &= \ T(\tpl x',f'(\tpl z')) \gets B'.
\end{align*}
Then $\tpl x$ and $\tpl x'$ have equal length.
Note that $\tpl x$ and $\tpl z$ as well as
$\tpl x'$ and $\tpl z'$ may have variables in common.

Our aim is to show that oid-equivalence between
arbitrary sifo CQs $Q$ and $Q'$
can be reduced to the case where the heads
$$
T(\tpl x,f(\tpl z)) \quad\mbox{and}\quad T(\tpl x',f'(\tpl z'))
$$
have identical arguments, that is,
where $\tpl x = \tpl x'$ and $\tpl z = \tpl z'$.

As a first lemma we state that rearranging the creation variables
of a query does not affect oid-equivalence.

\begin{lemma}[Rearranging creation variables]
  \label{lem:rearrc}
Let $Q$ be a sifo CQ written as above.
Let $\tpl u$ be a tuple with exactly the same variables as $\tpl z$,
but possibly with different repetitions and a different  ordering,
and let $g$ be a function symbol whose arity is equal
to the length of $\tpl u$.
Then the sifo CQ $P = T(\tpl x,g(\tpl u)) \gets B$ is oid-equivalent to $Q$.
\end{lemma}
\begin{proof}
Let $I$ be an instance.  We define an oid isomorphism from $Q(I)$
to $P(I)$ as follows.
Any oid $o$ in $Q(I)$ is of the form $f(\alpha(\tpl z))$ for some
matching $\alpha\colon B\to I$;
we define $\rho(o) := g(\alpha(\tpl u))$.  This is well-defined,
i.e., independent of the choice of $\alpha$.  Indeed, if
the data terms $f(\alpha_1(\tpl z))$ and
$f(\alpha_2(\tpl z))$ are equal, then the tuples
$\alpha_1(\tpl z)$ and
$\alpha_2(\tpl z)$ are equal,
which implies that $\alpha_1$ and $\alpha_2$
agree on every variable appearing in $\bar z$.  Since exactly the
same variables appear in $\bar u$, also the tuples
$\alpha_1(\tpl u)$ and
$\alpha_2(\tpl u)$ are equal,
whence
$g(\alpha_1(\tpl u))=
g(\alpha_2(\tpl u))$.

That $\rho : \oids{Q(I)} \to \oids{P(I)}$
is injective is shown by an analogous argument.
The surjectivity of $\rho$, as well as the equality
$\rho(Q(I)) = P(I)$, are clear.
\end{proof}

By the above lemma, we can remove all duplicates from $\bar z$
and $\bar z'$ in the heads of $Q$ and $Q'$, respectively.
So, from now on we may assume $\bar z$ and $\bar z'$ have no
duplicates.

In the following,
let $Z$ equal the set of variables occurring in $\bar z$,
let $X$ equal the set of variables occurring in $\bar x$, and let $Z'$
and $X'$ be defined similarly.

We next show that two sifo CQs can only be oid-equivalent if they
have identical patterns of distinguished variables, up to renaming.

\begin{lemma}[Renaming distinguished variables]
  \label{lem:rearrd}
If $Q$ and $Q'$ are oid-equivalent, then there exists a bijective
variable mapping $\sigma : X \to X'$ such that $\sigma(\bar x)=\bar x'$.%
\end{lemma}
\begin{proof}
Certainly, if $Q$ and $Q'$ are oid-equivalent, then the conjunctive
queries
$Q_0 = T_0(\bar x) \gets B$ and
$Q_0' = T_0(\bar x')\gets B'$,
where $T_0$ is a new predicate symbol, are equivalent.
So, there are homomorphisms $h\col Q_0 \to Q_0'$ and $h'\col Q_0' \to Q_0$.
In particular, $h(\bar x)=\bar x'$ and $h'(\bar x')=\bar x$.
We define $\sigma$ to be the restriction of $h$ to $X$.
The claim $\sigma(\bar x) = \bar x'$ and the
surjectivity of $\sigma$ are then clear.
So it remains to show that $\sigma$ is injective.
Thereto, consider $h'(\sigma(\bar x)) = h'(h(\bar x)) = h'(\bar x') = \bar x$.
We see that $h'\circ \sigma$ is the identity on $X$ and thus injective.
Hence, $\sigma$ must be injective as well.
\end{proof}

By the above lemma, if there does \emph{not} exist a renaming
$\sigma$ as in the lemma, certainly $Q$ and $Q'$ are not
oid-equivalent.  If there exists such a renaming, then by
renaming the variables in one of the two queries, we can now
assume without loss of generality that $\bar x= \bar x'$ and in
particular that $X=X'$.

The next step is to show that oid-equivalent queries must
have the same distinguished
variables among the creation variables, that is, $X\cap Z=X \cap Z'$.

\begin{lemma}[Distinguished creation variables]
  \label{lem:distinguished-creation-vars}
If $X\cap Z \neq X \cap Z'$, then $Q$ and $Q'$ are not oid-equivalent.
\end{lemma}
\begin{proof}
Either there exists some $x \in X\cap Z$ but not in $Z'$ or vice versa.
By symmetry we may assume the first possibility.

We construct an instance $I$ from $B'$.  In doing this, to keep
our notation simple, we consider the variables in $B'$ to be
constants.  The instance $I$ is obtained from $B'$ by
duplicating $x$ to some new element $x_2$.  Formally, consider
the mapping $d$ on $\var {B'}$ that is the identity everywhere except
that $x$ is mapped to $x_2$; then $I = B' \cup d(B')$.

First, let us look at $Q'(I)$.  Using the identity matching that
maps every variable to itself, we obtain the extended fact
$T(\tpl x,f'(\tpl z')) \in Q'(I)$.  Using the matching $d$
defined above, we obtain the extended fact $T(\tpl
x_2,\allowbreak f'(d(\tpl
z')))$ in $Q'(I)$.  Here, $\tpl x_2$ denotes $d(\bar x)$, i.e.,
$\bar x_2$ is obtained from $\tpl x$ by replacing $x$ with $x_2$.
Since $x$ does not belong to $Z'$, we have $d(\tpl z') =
\tpl z'$, so $T(\tpl x_2,f'(\tpl z')) \in Q'(I)$.

On the other hand, in $Q(I)$ consider any two extended facts
$T(\alpha_1(\tpl x),\allowbreak f(\alpha_1(\tpl z)))$ and
$T(\alpha_2(\tpl x),\allowbreak f(\alpha_2(\tpl z)))$,
with matchings $\alpha_1\col B \to I$ and $\alpha_2\col B \to I$,
such that $\alpha_1(\tpl x)=\tpl x$ and
$\alpha_2(\tpl x) = \tpl x_2$.
Then in particular $\alpha_1(x)=x$ and $\alpha_2(x)=x_2$.
Since $\alpha_1$ and $\alpha_2$ differ on $x$, and $x$ is in $Z$,
also $\alpha_1(\tpl z)$ and $\alpha_2(\tpl z)$ are different.
Hence, the two last components $f(\alpha_1(\tpl z))$ and $f(\alpha_2(\tpl z))$
are different.  Thus, we see that in $Q(I)$ it is impossible
to have two extended atoms $T(\tpl x,o)$ and $T(\tpl x_2,o)$ with the
same oid $o$.
But we have seen this is possible in $Q'(I)$,
so $Q(I)$ and $Q'(I)$ are not oid-isomorphic
and $Q$ and $Q'$ cannot be oid-equivalent.
\end{proof}

By the above Lemma we now assume $X\cap Z=X\cap Z'$.
The last step is to show that $Z\setminus X$ and $Z'\setminus X$,
the sets of non-distinguished creation variables,
need to have the same cardinality.

\begin{lemma}[Non-distinguished creation variables]
  \label{lem:non-distinguished-creation-vars}
If $Z\setminus X$ and $Z'\setminus X$ have different cardinality
then $Q$ and $Q'$ are not oid-equivalent.
\end{lemma}

\begin{proof}
As in the proof of Lemma~\ref{lem:distinguished-creation-vars},
we consider $B$ as an instance,
viewing variables as constants.

Let $k$ and $k'$ be the cardinalities of $Z\setminus X$ and
$Z\setminus X'$, respectively. By symmetry we may assume that $k
> k'$.  Now, for any natural number $n$, let $I_n$ be the
instance obtained from $B$ by independently multiplying each variable $z \in
Z\setminus X$ into $n$ fresh copies $z^{(1)}, \dots, z^{(n)}$.
Formally, for any function $d:Z \setminus X \to \{1,\dots,n\}$,
let $\hat d$ be the valuation on $\var B$ that maps each $z \in
Z-X$ to $z^{(d(z))}$ and that is the identity on all other
variables.  Then $$ I_n = \bigcup_{d:Z\setminus X\to\{1,\dots,n\}} \hat
d(B). $$

There are $n^k$ different functions $d:
Z-X\to\{1,\dots,\allowbreak n\}$.
Each corresponding valuation $\hat d$ is a matching of $B$ in
$I_n$; all these matchings are the identity on $\bar x$ but are
pairwise different on $\bar z$.  Thus there are at least $n^k$
different extended facts in $Q(I_n)$ of the form $T(\bar x,o)$.

On the other hand, consider any set $S$ of valuations from $X \cup Z'$
to $\adom{I_n}$ that are pairwise different on $Z'\setminus X$ but that all
agree on $X$.  The cardinality of $Z'\setminus X$ is
$k'$. The cardinality of $\adom{I_n}$ is
$O(n)$ (although the cardinality of $I_n$ itself is larger).
Hence, such a set $S$ can be of cardinality at most $O(n^{k'})$.
Consequently, since $k>k'$, for $n$ large enough, $Q'(I_n)$ cannot
possibly contain $n^k$ different extended facts of the form
$T(\bar x, o)$.  But we saw that this is possible in $Q(I_n)$.
So, $Q(I_n)$ and $Q'(I_n)$ are not oid-isomorphic and $Q$ and
$Q'$ cannot be oid-equivalent.
\end{proof}

By the above lemma, and after renaming the variables in $Z'\setminus X$ and
reordering the variables in $\bar z'$, we may
now indeed assume that $\bar z$ and $\bar z'$ are identical.

\subsection{Characterization of oid-equivalence}
\label{sec:charoid}

According to the results of the preceding subsection,
we are now given two sifo CQs as follows:
\begin{align}
Q &\ = \ T(\tpl x,f(\tpl z)) \gets B
  \label{eqn:normal-form-one}\\
Q'&\ = \ T(\tpl x,f'(\tpl z)) \gets B'.
  \label{eqn:normal-form-two}
\end{align}
Note that $Q$ and $Q'$ have identical tuples $\bar x$ and $\bar
z$ of distinguished and
creation variables; moreover, $\bar z$ contains no variable more
than once.
As before, we denote the sets of distinguished and creation variables
as $X$ and $Z$, respectively.

We will show that $Q$ and $Q'$ are oid-equivalent if and only if
there are homomorphisms between $B$ and $B'$ in both directions
that (i) keep $\tpl x$ fixed and (ii) possibly permute the
variables in $\tpl z$.  To make this formal, we associate to each
query a classical CQ without function symbols.

\begin{definition}
Fix a new relation symbol $\fltg T$ of arity the sum of the
lengths of $\bar x$ and $\bar z$.  The \emph{flattening} of $Q$
is the query $\fltg Q = \fltg T(\tpl x,\tpl z) \gets B$.
The query $\fltg Q'$ is defined similarly.
\end{definition}

Let $\pi$ be a permutation of the set $Z\setminus X$.
We extend $\pi$ to $\var B$ by defining it to be the identity outside $Z-X$.
We now define $\fltg Q^\pi$ to be the conjunctive query
obtained from $\fltg Q$ by permuting the variables in $\tpl z$,
that is
$$
 \fltg Q^\pi \ = \ \fltg T (\bar x,\pi(\bar z)) \gets B.
$$

This notion allows us to formulate the following natural sufficient condition
for oid-equivalence.

\begin{proposition}
  \label{theo:containment-modulo-permutation-implies-oid-equivalence}
If there exists a permutation $\pi$ of $Z-X$ such that $\fltg
Q^\pi$ and $\fltg Q'$ are equivalent, then $Q$ and $Q'$ are
oid-equivalent.
\end{proposition}

\begin{proof}
Let $I$ be an instance.  We define an oid isomorphism $\rho$
from $Q(I)$ to $Q'(I)$ as follows.
Any oid $o$ in $Q(I)$ is of the
form $f(\alpha(\bar z))$ for some matching $\alpha\col B \to I$;
we define $\rho(o) := f'(\alpha(\pi(\bar z)))$.
This is well-defined, i.e., independent of the choice of $\alpha$.
Indeed, if the data terms $f(\alpha_1(\bar z))$ and $f(\alpha_2(\bar z))$ are equal,
then the tuples $\alpha_1(\bar z)$ and $\alpha_2(\bar z)$ are equal,
and consequently the permuted tuples
$\alpha_1(\pi(\bar z))$ and $\alpha_2(\pi(\bar z))$ are equal.
Hence, $f'(\alpha_1(\pi(\bar z)))=f'(\alpha_2(\pi(\bar z)))$.

The injectivity of
$\rho : \oids{Q(I)} \to \oids{Q'(I)}$ is shown by an analogous
argument.
The surjectivity of $\rho$,
and the equality $\rho(Q(I)) = Q'(I)$,
follow readily from the equality $\fltg Q^\pi(I)=\fltg Q'(I)$.
\end{proof}

We next prove that the sufficient condition given by the above
Proposition is actually also necessary for oid-equivalence.
The key idea for proving this is to show that
oid-equivalence of sifo CQs depends only on the \emph{number} of oids
generated for any binding of the distinguished variables.

\renewcommand{\sharp}{\#}

Formally, for any instance $I$ and any tuple $\bar c$ of elements
from $\adom I$, we define
$$
  \sharp_{\tpl c}(Q,I) := \# \set{o \mid T(\tpl c,o) \in Q(I)},
$$
that is, $\sharp_{\tpl c}(Q,I)$ denotes the number of distinct
oids $o$ that occur together with $\tpl c$ in $Q(I)$.
We will show that $Q$ and $Q'$ are oid-equivalent
if and only if
$\sharp_{\tpl c}(Q,I) = \sharp_{\tpl c}(Q',I)$ for all instances $I$
and tuples $\tpl c$.  The only-if direction of this statement is
obvious, but the if-direction is not so obvious.

For our proof, we rely on work by Cohen \cite{cohen_multiplicities}
who studied queries with multiset variables that are evaluated under
so-called combined semantics,
a semantics that combines set and multiset semantics.
Cohen characterized equivalence of such queries in terms of homomorphisms.

Queries with multiset variables (MV queries) have the form
$Q_0,M$ where $Q_0$ is a standard CQ and $M$ is some set of
variables of $Q_0$ that do not appear in the head of $Q_0$.
The elements of $M$ are
called the multiset variables.  Evaluating an MV query $Q_0,M$ on
an instance $I$ results in a multiset (bag) of facts, where
the number of times a fact occurs is related to the number of
different possible assignments of values to the multiset variables.

Let us define the combined semantics formally.  Let $Q_0$ be of
the form $H_0 \gets B_0$
and let $I$ be an input instance.  Recall that
$Q_0(I)$ according to the classical semantics equals
$$ \{ \alpha(H_0) \mid \alpha : B_0 \to I \}. $$
Let $W$ be the set of variables appearing in $H_0$.
Then the result of evaluating the MV query $Q_0,M$ on instance $I$ is
defined to be the multiset with ground set $Q_0(I)$, where for
each fact $e \in Q_0(I)$, the multiplicity of $e$ in the multiset
is defined to be
$$ \# \{\gamma|_M \mid \text{$\gamma : B_0 \to I$ and
$\gamma(H_0) = e$} \}. $$
That is, given a fact $\alpha(H_0) \in Q_0(I)$, there may be many
different matchings $\gamma$ that agree with $\alpha$ on $H_0$.
The multiplicity of $\alpha(H_0)$ is defined to be not the
\emph{total} number of different such matchings $\gamma$, but
rather the number of different \emph{restrictions} one obtains
when restricting these matchings $\gamma$ to
$M$.\footnote{The motivation for MV queries was to model the semantics of
positive SQL queries with nested \texttt{EXISTS} subqueries.
While queries under standard SQL semantics return multisets of tuples,
only the relations mentioned in the top level SQL block contribute to
the multiplicities of answers, whereas relations mentioned in the subquery
do not.}

Two MV queries are \emph{equivalent} if they evaluate to the same
multiset on every input instance.  Equivalence of MV queries can be
characterized using the notion of multiset-homomorphism
\cite{cohen_multiplicities}.  A \emph{multiset-ho\-mo\-mor\-phism} from
MV query $Q_0,M$ to MV query $Q'_0,M'$ is a homomorphism $h : Q_0
\to Q_0'$ such that $h$ is injective on $M$ and $h(M) \subseteq
M'$.  Cohen showed the following:
\begin{theorem}[\cite{cohen_multiplicities}, Thm~5.3]
\label{cohentheorem}
Two MV queries are equivalent if and only
if there are multiset homomorphisms between them in both
directions.
\end{theorem}

To leverage this result on MV equivalence,
we associate two MV queries to our given sifo CQs in the
following way.
\begin{definition}
Fix a new relation symbol $T_0$ of arity the length of $\bar x$.
The MV queries $\tilde Q$ and $\tilde Q'$ are defined as
$Q_0,(Z-X)$ and $Q_0',(Z-X)$ respectively, where
\begin{align*}
Q_0 &= T_0(\bar x) \gets B \\
Q_0' &= T_0(\bar x) \gets B'
\end{align*}
\end{definition}

The following proposition now relates oid-equiv\-a\-lence
to MV-equivalence:

\begin{proposition}
  \label{theo:oid-equivalence-implies-combined-equivalence}
If $Q$ and $Q'$ are oid-equivalent, then the MV queries
$\tilde Q$ and $\tilde Q'$ are equivalent.
\end{proposition}
\begin{proof}
Let $I$ be an instance.  We must show that the multisets
$\tilde Q(I)$ and $\tilde Q'(I)$ are equal.
Since $Q$ and $Q'$ are oid-equivalent,
the ground sets $Q_0(I)$ and
$Q_0'(I)$ of $\tilde Q(I)$ and $\tilde Q'(I)$ are already equal.
We must show that the element multiplicities are the same as well.

Let $T_0(\bar c)$ be an arbitrary element of $Q_0(I)$.
By the semantics of oCQs,
we have the following equalities:
\begin{align*}
\#_{\bar c} (Q,I) & =
\# \{ \gamma|_{X \cup Z} \mid
\text{$\gamma : B \to I$ and
$\gamma(\bar x)=\bar c$}\} \\
\#_{\bar c} (Q',I) & =
\# \{ \gamma|_{X \cup Z} \mid
\text{$\gamma : B' \to I$ and
$\gamma(\bar x)=\bar c$}\}
\end{align*}
Since $Q(I)$ and $Q'(I)$ are oid-isomorphic, the left-hand sides
of the above two equalities are equal.  Hence, the right-hand
sides are equal as well.  But these are precisely the
multiplicities of $T_0(\bar c)$ in $\tilde Q(I)$ and $\tilde
Q'(I)$ respectively.
\end{proof}

The following proposition further relates
MV equivalence to equivalence of the flattenings up to permutation:

\begin{proposition}
   \label{theo:combined-equivalence-implies-containment-modulo-permutation}
If the MV queries $\tilde Q$ and $\tilde Q'$ are equivalent
then there exists a permutation $\pi$ of $Z-X$ such that
$\fltg Q^\pi$ and $\fltg Q'$ are equivalent.
\end{proposition}
\begin{proof}
By Theorem~\ref{cohentheorem},
there exist a multiset homomorphism
$h$ from $\tilde Q$ to $\tilde Q'$, and
a multiset homomorphism
$h'$ from $\tilde Q'$ to $\tilde Q$.
Since Theorem~\ref{cohentheorem} also implies
that $h$ is injective on $Z-X$ and that $h(Z-X)\subseteq Z-X$, we can
conclude that $h$ acts as a permutation on $Z-X$. Moreover, $h$
is the identity on $X$.  The same two properties hold for $h'$.

Now put $\pi = (h|_{Z-X})^{-1}$.  Then $h : \fltg Q^\pi \to \fltg
Q'$.  So it remains to find a homomorphism $h'' : \fltg Q' \to
\fltg Q^\pi$.  Thereto, note that $h'h$ acts as a permutation
on $Z-X$.  Since $Z-X$ is finite, there exists a nonzero natural
number $m$ such that $(h'h)^m$ is the identity on $Z-X$.
Equivalently, $(h'h)^{m-1}h'$ equals $\pi$ on $Z-X$.  We conclude
that $(h'h)^{m-1}h'$ is the desired homomorphism $h''$.
\end{proof}

We summarize the
three preceding Propositions in the following.

\begin{theorem}
  \label{theoroid}
Consider two sifo CQs
\begin{align*}
Q &\ = \ T(\tpl x,f(\tpl z)) \gets B \\
Q'&\ = \ T(\tpl x,f'(\tpl z)) \gets B'
\end{align*}
where $Q$ and $Q'$ have identical tuples $\bar x$ and $\bar
z$ of distinguished and
creation variables, and where $\bar z$ contains no variable more
than once.  Denote the sets of distinguished and creation variables
by $X$ and $Z$, respectively.

The following are equivalent:
\begin{enumerate}
\item The sifo CQs $Q$ and $Q'$ are oid-equivalent;
\item The MV queries $\tilde Q$ and $\tilde Q'$ are equivalent;
\item There is a permutation $\pi$ of $Z-X$ such that the
classical CQs $\fltg Q^\pi$ and $\fltg Q'$ are equivalent.
\end{enumerate}
\end{theorem}

\subsection{Computational complexity}

The results of this section imply the following:

\begin{corollary}
Testing oid-equivalence of sifo CQs is NP-complete.
\end{corollary}
\begin{proof}
Assume given sifo CQs $Q$ and $Q'$ with the same head predicate:
\begin{align*}
Q \ &= \ T(\tpl x,f(\tpl z)) \gets B
\\
Q'\ &= \ T(\tpl x',f'(\tpl z')) \gets B'.
\end{align*}
Let $X$, $X'$, $Z$ and $Z'$ denote the sets of variables
occurring in $\bar x$, $\bar x'$, $\bar z$ and $\bar z'$,
respectively.

To test oid-equivalence, we begin by removing duplicates in $\bar
z$ and $\bar z'$, as justified by Lemma~\ref{lem:rearrc}.
Note that $\bar x$ and $\bar x'$ have the same length $k$, because of
the fixed arity of $T$.  So we can write $\bar x=x_1,\dots,x_k$
and $\bar x'=x'_1,\dots,x'_k$.  Consider the mapping $\sigma =
\{(x_1,x'_1),\allowbreak \dots, \allowbreak
(x_k,x'_k)\}$.  We test if $\sigma$ is a bijection
from $X$ to $X'$; if not, then $Q$ and $Q'$ are not oid-equivalent by
Lemma~\ref{lem:rearrd}.
If $\sigma$ is a bijection, we can safely replace every variable
$x'$ in $X'$ by $\sigma^{-1}(x')$, which yields a sifo CQ that is
oid-equivalent to $Q'$.  Hence, from now on we may assume that
$\bar x=\bar x'$ and in particular $X=X'$.

Next, we test whether $X \cap Z = X \cap Z'$ and whether
$Z-X$ and $Z'-X$ have the same
cardinality; if one of the two tests fails then $Q$ and $Q'$ are
not oid-equivalent by Lemmas
\ref{lem:distinguished-creation-vars}
and
\ref{lem:non-distinguished-creation-vars}.  Otherwise, we can
rename the variables in $Z'-X$, so that we may
assume that $\bar z=\bar z'$.

We are now left in the situation where $Q$ and $Q'$ are in the
general forms (\ref{eqn:normal-form-one})
and (\ref{eqn:normal-form-two}) from
Subsection~\ref{sec:charoid}, to which Theorem~\ref{theoroid}
applies.  By the third statement of
this theorem we can test oid-equivalence of $Q$ and
$Q'$ in NP by guessing a permutation $\pi$ and two homomorphisms
between $\fltg Q^\pi$ and $\fltg Q'$ in both directions.

NP-hardness follows immediately because the problem has equivalence of
classical CQs as a special case, which is well known to be
NP-hard.  Indeed, oid equivalence of sifo CQs $Q$ and $Q'$ in the
special case where the creation functions are nullary, amounts to
classical equivalence when we ignore the function terms in the
heads.
\end{proof}


\section{Logical entailment of sifo CQs interpreted as schema
mappings} \label{seclog}

Object-creating CQs, and sifo CQs in particular, can also be interpreted
alternatively as schema mappings rather than as queries.
Specifically, consider a sifo CQ $Q$ of the general form
$T(\bar x,f(\bar z)) \gets B$ over the database
schema $\schema S$.  Let $\bar v$ be the sequence of all
variables used in $B$.  Then we may view $Q$ as a second-order
implicational statement
over the augmented schema $\schema S \cup \{T\}$, as
follows:
$$ \exists f \forall \bar v (B \to H) $$
Here, $H$ is the head
and $B$ is conveniently used to stand for the conjunction of its elements.
Note that this formula is second-order because it existentially quantifies
a function $f$; we denote the above formula by $\sotgd Q$.
This formula belongs to the well-known class of second-order
tuple-generating dependencies (SO-tgds).
More specifically,
it is a \emph{plain} SO-tgd \cite{plainsotgd}.

Syntactically,
the plain SO-tgds coming from sifo CQs
in this manner form a restricted class of SO-tgds, defined by the
following restrictions:
\begin{itemize}
\item
Plain SO-tgd may consist of
multiple rules; sifo CQs consist of a single rule.
\item
The head of a plain SO-tgd may consist of multiple atoms;
the head of a sifo CQ consists of a single atom.  (This is similar
to GAV mappings \cite{Lenzerini02,ck_cacm_mappings}, although the
classical notion of GAV mapping does not use function symbols.)
\item
There is only one function symbol, which moreover can be applied
only once in the head.
\end{itemize}

When interpreting a sifo CQ $Q$ as an SO-tgd, the semantics becomes
that of a schema mapping.  Specifically, let $I$ be an instance
over $\schema S$, considered as a source instance, and let $J$ be
an instance over $\{T\}$, considered as a target instance.  Then
$(I,J)$ together form an instance over the augmented schema
$\schema S \cup \{T\}$.  Now we say that $(I,J)$ satisfies $Q$,
denoted by $(I,J) \models Q$, if the structure $(\adom I \cup \adom J,I,J)$
satisfies $\sotgd Q$ under the standard semantics of second-order
logic, using $\adom I \cup \adom J$ as the universe of the structure.

The following example and remark illustrate that the semantics of sifo CQs
as SO-tgds is quite different from their semantics as
object-creating queries.


\begin{table}
\centering
\begin{tabular}[t]{|ll|}
\multicolumn{2}{c}{Family} \\
\hline
beth & jones \\
ben & jones \\
eric & simpson \\
emma & smith \\
\hline
\end{tabular}
\qquad
\begin{tabular}[t]{|ll|}
\multicolumn{2}{c}{Family} \\
\hline
beth & jones \\
ben & jones \\
eric & jones \\
emma & jones \\
\hline
\end{tabular}
\caption{Instances $J_1$ and $J_2$ from
Example~\ref{exweirdfamily}.}
\label{tabtwosolutions}
\end{table}

\begin{table}
\centering
\begin{tabular}[t]{|ll|}
\multicolumn{2}{c}{Family} \\
\hline
beth & jones \\
ben & murphy \\
eric & simpson \\
emma & smith \\
\hline
\end{tabular}
\caption{Instance $J_3$ from
Example~\ref{exweirdfamily}.}
\label{tabnosolution}
\end{table}

\begin{example}
\label{exweirdfamily}
Let us consider again our query from Example 1. As we have mentioned in the Introduction, 
we can now write it as an SO-tgd as follows:
$$
\exists f \forall x \forall y \forall c (\mathit{Mother}(c, x)
\land \mathit{Father} (c,y) \to \mathit{Family}(c,f(x,y)))
$$

Take the instance $I$ consisting of the Mother
and Father facts listed in Table~\ref{exfamily}, and take the
instances $J_1$ and $J_2$ consisting of the Family facts listed
in Table~\ref{tabtwosolutions} left and right respectively.  Then
both pairs $(I,J_1)$ and $(I,J_2)$ satisfy the SO-tgd.  For $J_1$
this is witnessed by the following function $f$:
$$ \begin{array}{lll}
x & y & f(x,y) \\
\hline
\rm anne & \rm adam & \rm jones \\
\rm claire & \rm carl & \rm simpson \\
\rm diana & \rm carl & \rm smith
\end{array} $$
For $J_2$ this is
witnessed by the function that simply maps everything to jones.

In contrast, for $J_3$ consisting of the Family
facts listed in Table~\ref{tabnosolution}, the pair $(I,J_3)$
does not satisfy the SO-tgd.  Indeed, suppose there would exist a
function $f$ witnessing the truth of the formula on $(I,J_3)$.
Since beth has anne as
mother and adam as father,
the fact $$ {\rm Family}({\rm beth}, f({\rm anne},{\rm adam})) $$
must belong to $J_3$.  The only Family-fact with beth in the
first position is $$ \rm Family(beth,jones), $$ so we conclude
$$f({\rm anne},{\rm adam}) = {\rm jones}.$$
Furthermore, since ben also has anne as mother and adam as
father, the fact
$${\rm Family}({\rm ben}, f({\rm anne},{\rm adam}))$$
must be in $J_3$.  The only Family-fact with ben in the
first position is $$\rm Family(ben,murphy),$$ however, so we must conclude
that
$$f({\rm anne},{\rm adam}) = {\rm murphy},$$ which is in
contradiction with the previous conclusion.
\end{example}

\begin{remark}
Note that, by the purely implicational nature of SO-tgds,
if $(I,J)$ satisfies an SO-tgd and $J \subseteq J'$, then also
$(I,J')$ satisfies the SO-tgd.  Hence, continuing the previous
example, for any instance $J'$ obtained by $J_1$ or $J_2$ by adding some
more Family-facts, the pair $(I,J')$ would still satisfy the
SO-tgd from the example.
\qed
\end{remark}

The above example and remark show that given a source instance
$I$, there are in general multiple possible target instances $J$
such that $(I,J) \models Q$.  This is in contrast to the
semantics of $Q$ as an oCQ, where $Q(I)$ is an extended instance
that is uniquely defined.  Still, there is a connection
between the oCQ semantics and the SO-tgd semantics.  Specifically,
$Q(I)$ can be viewed as a target instance in a canonical manner,
using \emph{oid-to-constant assignments} (oc-assignments for
short) defined as follows.

\begin{definition}
Let $I$ be a source instance and let $J$ be an extended instance
over $\{T\}$ such that $\consts J \subseteq \adom I$.
An \emph{oc-assignment}
for $J$ with respect to $I$ is an injective mapping
$\rho : \oids J \to \dom$ so that the
image of $\rho$ is disjoint from $\adom I$.
\end{definition}
Thus, $\rho$ assigns
to each non-constant data term from $J$ a different constant that
is not in $\adom I$.

We now observe the following obvious property giving a connection between
the oCQ semantics and the SO-tgd semantics:

\begin{proposition} \label{propoca}
Let $I$ be a source instance and let $\rho$ be an oc-assignment
for $Q(I)$ with respect to $I$.  Then $(I,\rho(Q(I))) \models Q$.
\end{proposition}

In fact, $Q(I)$ corresponds to what Fagin et al.\ \cite{sotgd}
call the \emph{chase of $I$ with $\sotgd Q$}.

\subsection{Nested dependencies} \label{secnested}

We have introduced sifo CQs as a restricted class of plain
SO-tgds.  But actually, sifo CQs can also be considered as a
restricted form of so-called nested tgds \cite{miller_invention}.
Thereto, consider again a sifo CQ of the general form $T(\bar
x,f(\bar z)) \gets B$.  Let $\bar u$ be the sequence of all
variables from $B$, except for the creation variables (the
variables from $\bar z$).  Furthermore, let $w$ be a fresh variable not
occurring in $B$, and let $H'$ be the atom $T(\bar x,w)$.
We can now associate to $Q$ the
following implicational statement, denoted by $\ntgd Q$:
$$ \forall \bar z \exists w \forall \bar u (B \to H') $$
Note that $\ntgd Q$ is now a first-order formula, but
it is clear that $\ntgd Q$ is logically equivalent to $\sotgd Q$.
Hence, the schema mappings arising from sifo CQs are not
essentially second-order in nature.

\subsection{Logical entailment}

In Section~\ref{secoidequiv} we have shown that equivalence of
sifo CQs as object-creating queries is decidable.  Now that we
have seen that sifo CQs can also be given a semantics as schema
mappings, we may again ask if equivalence under this alternative
semantics is decidable.  The answer is affirmative; we have seen
in the previous subsection that sifo CQ mappings belong to the
class of nested dependencies, and logical implication of nested
dependencies has recently been shown to be decidable
\cite{pichler_nestedep}.  When this general implication test for
nested dependencies is applied specifically to sifo CQ schema
mappings, it can be implemented in non-deterministic polynomial
time.  Hence, logical entailment (and also logical equivalence)
of sifo CQ schema mappings is NP-complete.

In the present section, we present a specialized logical
entailment test for sifo CQ schema mappings which is much simpler and
more elegant, and provides more insight in the problem by
relating it to testing implication of a join dependency by a
conjunctive query (Theorem~\ref{theorentail}).   Interestingly,
there is a striking correspondence between the general
implication test when applied to sifo CQs, and the strategy we
use to prove our theorem.  An in-depth comparison will be
given in Section~\ref{secdiscuss}, after we have stated the
Theorem formally and have seen its proof.

Formally, given two schema mappings $\M$ and $\M'$ from a source
schema $\schema S$ to a target schema $\{T\}$, we say that
\emph{$\M$ logically entails $\M'$} if the following implication
holds for every instance $I$ over $\schema S$ and every instance
$J$ over $\{T\}$: $$ \text{$(I,J)$ satisfies $\M$} \quad
\Rightarrow \quad \text{$(I,J)$ satisfies $\M'$}. $$

Referring to the view of sifo CQs as SO-tgds introduced above, we
now define:
\begin{definition}
Let $Q$ and $Q'$ be two sifo CQs with the same head predicate, and
over the same database sche\-ma.
We say that \emph{$Q$ logically entails $Q'$} if $\sotgd{Q}$
logically entails $\sotgd{Q'}$.
\end{definition}

\begin{table}
$$
\begin{array}[t]{|lll|}
\multicolumn{3}{c}{I} \\
\hline
a_1 & b & c \\
a_2 & b & c \\
\hline
\end{array}
\qquad
\begin{array}[t]{|ll|}
\multicolumn{2}{c}{J} \\
\hline
a_1 & d_1 \\
a_2 & d_2 \\
\hline
\end{array}
$$
\caption{Instances used in Example~\ref{exentail}.}
\label{tabnotentail}
\end{table}

\begin{example}
\label{exentail}
Recall the sifo CQs $Q$ and $Q'$ from
Example~\ref{exnotoidequiv}:
\begin{align*}
Q & = T(x,f(y)) \gets R(x,y,z) \\
Q' & = T(x,f'(x,y)) \gets R(x,y,z)
\end{align*}
It is clear that
$Q$ logically entails $Q'$.  Indeed, if there exists a function
$f$ witnessing the truth of $\sotgd Q$, then we can easily define
a function $f'$ witnessing the truth of $\sotgd{Q'}$ by defining
$f'(x,y) := f(y)$.

Conversely, however, $Q'$ does not logically entail $Q$.  Indeed,
Table~\ref{tabnotentail} shows $(I,J)$ where $(I,J) \models Q'$
but $(I,J) \not \models Q$.
\end{example}

\begin{example}
\label{ex2entail}
Recall the sifo CQs $Q$ and $Q'$ from
Example~\ref{ex2notoid}:
\begin{align*}
Q & = T(x,f(x)) \gets R(x,y,z) \\
Q' & = T(x,f'(x,y,z)) \gets R(x,y,z)
\end{align*}
Although $Q$ and $Q'$ are not oid-equivalent, they are logically
equivalent: they logical entail each other.
The logical entailment of $Q'$ by $Q$ is again clear.
To see the converse direction, assume $f'$
witnesses the truth of $\sotgd{Q'}$.
Then we define $f(x)$ for
any $x$ as
follows: if there exists a pair $(y,z)$ such that $R(x,y,z)$ holds,
we fix one such pair $(y,z)$ arbitrarily and define $f(x) :=
f'(x,y,z)$.
If no such $y$ and $z$ exist, we may define $f(x)$ arbitrarily.  It
is now clear that this $f$ witnesses the truth of $\sotgd Q$.
\end{example}

\begin{example}
Consider the sifo CQs:
\begin{align*}
Q & = T(x,f(z_1)) \gets R(z_1,x),R(z_1,z_2) \\
Q' & = T(x,f'(z_1,z_2)) \gets R(z_1,x),R(z_1,z_2)
\end{align*}
Also here, $Q$ and $Q'$ logically entail each other.
The logical entailment of $Q'$ by $Q$ is again clear.
To see the converse direction, we can use a reasoning similar to
that used in Example~\ref{ex2entail}. Assume $f'$
witnesses the truth of $\sotgd{Q'}$.  Then we define $f(z_1)$ for
any $z_1$ as follows: if there exists $z_2$ such that $R(z_1,z_2)$
holds, we fix one such $z_2$ arbitrarily and define $f(z_1) :=
f'(z_1,z_2)$.  If no such $z_2$ exists, we may define $f(z_1)$
arbitrarily.  The function $f$ thus defined witnesses the truth of
$\sotgd Q$.

Note that the kind of reasoning used here and in Example~\ref{ex2entail}
does not work in the case of Example~\ref{exentail}.
In Theorem~\ref{theorentail} we will characterize
formally when this kind of
reasoning is correct.
\end{example}

Example~\ref{ex2entail} shows that logical equivalence (logical
entailment in both directions) does not imply oid-equivalence of
sifo CQs.  We will see in Theorem~\ref{oid2log} that
the other direction does hold.

\subsection{Join dependencies and tableau queries}

In our characterization of sifo CQ logical entailment we use a
number of concepts from classical relational database theory
\cite{ahv_book}, which we recall here briefly.

Recall that a \emph{relation scheme} is a finite set of elements
called \emph{attributes}.  It is customary to denote the union of
two relation schemes $X$ and $Y$ by juxtaposition, thus writing
$XY$ for $X \cup Y$.

A \emph{tuple} over a relation scheme $U$ is a function from $U$
to $\dom$.  A \emph{relation} over $U$ is a finite set of tuples
over $U$.

Let $t$ be a tuple over $U$ and let $X \subseteq U$.  The
restriction of $t$ to $X$ is denoted by $t[X]$.  The
\emph{projection} $\pi_X(r)$ of a relation $r$ over $U$ equals
$\set{t[X] \mid t \in r}$.

We now turn to tableau queries, which are an alternative
formalization of conjunctive queries so that the result of a
query is a set of tuples rather than a set of facts.  Let
$\schema S$ be a database schema, and let $B$ be a finite set of
atoms with predicates from $\schema S$, as would be the body of a
conjunctive query over $\schema S$.  Let $V = \var B$.  For any
$U \subseteq V$, the pair $(B,U)$ is called a \emph{tableau
query} over $\schema S$.  When applied to an instance $I$ over
$\schema S$, this tableau query returns a relation over $U$ in
the following manner.  Let $\Mat BI$ be the set of all matchings
of $B$ in $I$.  Using variables for attributes, $V$ can be viewed
as a relation scheme.  Under this view, every valuation on $V$ is
a tuple over $V$, and thus $\Mat BI$ is a relation over $V$.
We now define the result of $(B,U)$ on input $I$ to be
$\pi_U(\Mat BI)$.  This result is denoted by $(B,U)(I)$.

We finally recall join dependencies.  Let $t_1$ and $t_2$ be
tuples over the relation schemes $U_1$ and $U_2$, respectively.
If $t_1$ and $t_2$ agree on $U_1 \cap U_2$, the union $t_1 \cup t_2$
(where we take the union of two functions, viewed as sets of
pairs) is a well-defined tuple over the relation scheme $U_1U_2$.
The \emph{natural join} $r_1 \Join r_2$, for relations $r_1$ and
$r_2$ over $U_1$ and $U_2$, respectively, then equals $$ \{t_1 \cup
t_2 \mid t_1 \in r_1 \ \& \ t_2 \in r_2 \ \& \ t_1[U_1 \cap U_2]
= t_2[U_1 \cap U_2]\}. $$

Consider now any relation $r$ over some relation
scheme $U$.  Let $U_1$ and $U_2$ be subsets of $U$ (not
necessarily disjoint) such that
$U = U_1U_2$.  Then $r$
satisfies the \emph{join dependency (JD)} $U_1 \Join U_2$ if $r =
\pi_{U_1}(r) \Join \pi_{U_2}(r)$.  Note that the containment from
left to right is trivial, so one only needs to verify the
containment $\pi_{U_1}(r) \Join \pi_{U_2}(r) \subseteq r$.

The logical implication of JDs by tableau queries is well
understood and can be solved by the chase procedure
with NP complexity
\cite{klugprice,ahv_book}.  Formally, a tableau query $Q=(B,U)$ over
$\schema S$ is said to \emph{imply} a JD over $U$ if for every
instance $I$ over $\schema S$, the relation $Q(I)$ satisfies this
JD.

\subsection{Decidability of sifo CQ logical entailment}

We consider two sifo CQs $Q$ and $Q'$ with the same head predicate:
\begin{align*}
Q & = T(\bar x,f(\bar z)) \gets B \\
Q' & = T(\bar x',f'(\bar z')) \gets B'
\end{align*}

\begin{remark}
We assume $Q$ and $Q'$ to have their function symbol in the same
position in the head (here taken to be the last position).
This is justified because otherwise $Q$ could never logically
entail $Q'$.  In proof, suppose the function symbol in the head
of $Q'$ would not be in the last position.  Then we have a variable
$x'$ from $B'$ in the last position.  Now consider an instance $I$
such that both $Q(I)$ and $Q'(I)$ are nonempty.  (Such an
instance could be constructed by taking the disjoint union of $B$
and $B'$ and substituting constants for variables.) Let $\rho$
by an oc-assignment for $Q(I)$ with respect to $I$.  By
Proposition~\ref{propoca}, we have $(I,\rho(Q(I))) \models Q$.
In $\rho(Q(I))$, none of the elements in the last position of a
$T$-fact belongs to $\adom I$.  But then $(I,\rho(Q(I)))$ cannot
satisfy $Q'$.  Indeed, since $Q'(I)$ is nonempty, there is a
matching $\alpha':B'\to I$. In any $J'$ such that $(I,J')
\models Q'$, there needs to be a $T$-fact with $\alpha'(x')$ in
the last position, and $\alpha'(x') \in \adom I$.  We conclude
that $Q$ does not logically entail $Q'$.  
\qed
\end{remark}

In what follows we use $X$, $Z$ and $Z'$ to denote the sets of
variables appearing in the tuples $\bar x$, $\bar z$ and $\bar
z'$, respectively.

We establish:

\begin{theorem} \label{theorentail}
$Q$ logically entails $Q'$ if and only if there exists
a homomorphism $h : B \to B'$ satisfying the
following conditions:
\begin{enumerate}
\item $h(\bar x) = \bar x'$;
\item $h(X \cap Z) \subseteq Z'$;
\item Let $Y_h := h^{-1}(Z')$, i.e.,
$ Y_h = \{ y \in \var B \mid h(y) \in Z' \} $.
Then the tableau query $(B,XY_hZ)$ implies the join
      dependency $XY_h \Join Y_hZ$.
\end{enumerate}
\end{theorem}
\paragraph{Proof of sufficiency.}
Let $(I,J) \models Q$, witnessed by the function $f$.  We must
show $(I,J) \models Q'$.  This means finding a function $f'$
witnessing the truth of $\sotgd{Q'}$ in $(I,J)$.

Call any two matchings $\alpha_1,\alpha_2 \in \Mat BI$ equivalent
if they agree on $Y_h$.  This is denoted by $\alpha_1 \equiv
\alpha_2$.  Let $\rho$ be any function from $\Mat BI$ to $\Mat
BI$ with the two properties, first, that $\rho(\alpha) \equiv
\alpha$ and, second, that $\alpha_1 \equiv \alpha_2$ implies
$\rho(\alpha_1) = \rho(\alpha_2)$.  Thus, $\rho$ amounts to
choosing a representative out of each equivalence class.  We
denote the application of $\rho$ by subscripting, writing
$\rho(\alpha)$ as $\rho_\alpha$.

Let us define $f'$ as follows.
Take any matching $\beta : B' \to I$.
Then we put $f'(\beta(\bar z')) := f(\rho_{\beta \circ h}(\bar z))$.
To see that this is well-defined, recall that $h(Y_h) \subseteq Z'$.
Hence, $\beta_1(\bar z') = \beta_2(\bar z')$ implies that
$\beta_1 \circ h \equiv \beta_2 \circ h$,
so $\rho_{\beta_1\circ h} = \rho_{\beta_2 \circ h}$.

We now show that this interpretation of $f'$ satisfies the
requirements.  Specifically, let $\beta : B' \to I$ be a
matching.  We must show that $T(\beta(\bar x'),f'(\beta(\bar
z'))) \in J$.  Consider the valuations $\beta_1=\beta \circ h$
and $\beta_2=\rho_{\beta \circ h}$, both belonging to $\Mat BI$,
and viewed as tuples over the relation scheme $\var B$.  Since these
two tuples agree on $Y_h$, also the two restrictions
$\beta_1[Y_hX]$ and $\beta_2[Y_hZ]$ agree on $Y_h$.  Since $X
\cap Z \subseteq Y_h$, the union $\beta_1[Y_hX] \cup
\beta_2[Y_hZ]$ is a well-defined tuple over $XY_hZ$.  Since
$\pi_{XY_hZ}(\Mat BI)$ satisfies the JD $Y_hX \Join Y_hZ$, the
union belongs to $\pi_{XY_hZ}(\Mat BI)$.  Hence, there exists a
valuation $\gamma \in \Mat BI$ that agrees with $\beta \circ h$
on $X$, and with $\rho_{\beta \circ h}$ on $Z$.  Since $(I,J)
\models Q$, we have $T(\gamma(\bar x),f(\gamma(\bar z))) \in J$.
By the preceding, $\gamma(\bar x) = \beta(h(\bar x))$ and
$\gamma(\bar z) = \rho_{\beta \circ h}(\bar z) = g(\beta(\bar
z'))$.  We conclude that $T(\beta(\bar x'),g(\beta(\bar z'))) \in
J$ as desired.

\paragraph{Proof of necessity.}
Let $V'=\var{B'}$, and
let $n$ be the arity of $f$.
For each $l \in \{0,1,\dots,n\}$ and each $u \in V'\setminus Z'$
we introduce a fresh copy of $u$, denoted by $u^l$.
We say that this fresh copy is ``colored'' with color $l$.
For each variable $u \in Z'$, we simply define $u^l$ to be $u$ itself.
We say that the variables in $Z'$ are ``colored white''.

For any tuple of variables $\bar u=(u_1,\dots,u_p)$ in $V'$,
we denote the tuple $(u_1^l,\dots,u_p^l)$ by $\bar u^l$.
In this tuple, all variables are colored $l$ or white.
We then define
$B'^l = \set{R(\bar u^l) \mid R(\bar u) \in B'}$
and view it as an instance,
i.e., the variables $u^l$ are considered to be constants.

Now define the instance $I = \bigcup_{l=0}^n B'^l$, and construct
the instance $J = Q(I)$.  By Proposition~\ref{propoca}, $(I,J)
\models Q$, where we omit the oc-assignment for the sake of
clarity.  Since $Q$ logically entails $Q'$, also $(I,J) \models
Q'$.  Hence, there exists a function $f'$ such that for each
color $l$, using the matching $\id^l: B' \to I$,  $u \mapsto u^l$,
the fact $T(\bar x'^l,f'(\bar z'^l)) = T(\bar x'^l,f'(\bar z'))$
belongs to $J$.

Since $J=Q(I)$,
we have $f'(\bar z') = f(\bar w)$ for some tuple $\bar w$ of
colored variables in $V'$.
Since the arity of $f$ is $n$ and there are $n+1$ distinct colors,
some color does not appear in $\bar w$.
Without loss of generality we may assume that this is the color 0.

Let us now focus on the fact $T(\bar x'^0,f(\bar w))$ in $J$.
Like any $T$-fact in $J$, this fact has been produced
by some matching $k\colon B \to I$
such that $T(\bar x'^0,f(\bar w)) = k(T(\bar x,f(\bar z)))$,
so
\begin{enumerate}
\item[(a)] $k(\bar x)=\bar x'^0$ and
\item[(b)] $k(\bar z)=\bar w$.
\end{enumerate}

Let $s$ denote the mapping that removes colors, i.e.,
$s(u^l)=u$ for every $u \in V'$ and every $l \in \{0,1,\dots,n\}$.
Since $s(I) \subseteq B'$, we have a homomorphism
$s \circ k\colon B \to B'$.
We now define $h := s \circ k$
and show that it satisfies the conditions required by the Theorem.
The first condition is clear since
$h(\bar x) = s(k(\bar x)) = s(\bar x'^0) = \bar x'$.

For the second condition, let $x \in X \cap Z$.
By (a), $k(x)$ is colored 0 or white.
By (b), $k(x)$ is colored non-zero or white.
Hence, $k(x)$ is colored white, i.e., $k(x) \in Z'$,
so $h(x)=s(k(x))=k(x) \in Z'$ as desired.

Finally, to show that $(B,XY_hZ)$ implies $XY_h \Join Y_hZ$ we
must establish the query containment
$$(B,XY_h) \Join (B,Y_hZ) \subseteq (B,XY_hZ).$$
Treating tableau queries as conjunctive queries, and using the
well-known containment criterion for conjunctive queries,
this amounts to showing the existence of a certain homomorphism.
More specifically, we express the query $(B,XY_h) \Join (B,Y_hZ)$
by the conjunctive query with the body $B_2 = B_0 \cup B_1$
defined as follows.
The body $B_0$ is obtained from $B$
by replacing each variable $u$ not in $Y_h$ by a fresh copy $u^0$.
For each $u \in Y_h$ we define $u^0$ simply as $u$ itself.
The body $B_1$ is obtained from $B$
by replacing each variable not in $Y_h$ by a fresh copy $u^1$.
Again, for each $u \in Y_h$ we define $u^1$ simply as $u$ itself.
To show the containment, we now must find a homomorphism $m$
from $B$ to $B_2$ such that each $u\in X - Y_h$ is mapped to $u^0$;
each $u \in Y_h$ is mapped to $u$;
and each $u \in Z-Y_h$ is mapped to $u^1$.

Thereto, we define the following mapping $m$:
\begin{itemize}
\item
if $k(u)$ is colored 0, then $m(u):=u^0$;
\item
if $k(u)$ is colored $l$ for some $l>0$, then $m(u):=u^1$;
\item
if $k(u)$ is colored white, then $m(u):=u$.
\end{itemize}
Let us verify that $m\colon B \to B_2$ is a homomorphism.
Consider an atom $R(\bar u)$ in $B$;
we must show $R(m(\bar u)) \in B_2$.
Since $k\colon B \to I$, we know that $R(k(\bar u)) \in I$.
By definition of $I$, this means that $R(k(\bar u)) = R(\bar v^l)$
for some atom $R(\bar v)$ in $B'$ and some color $l$.
So, for each variable $u$ in $\bar u$, the color of $k(u)$ is either $l$ or white.
We now distinguish two cases.
\begin{itemize}
\item
If $k(u)$ is colored white, then
$h(u)=k(u) \in Z'$ so $u \in Y_h$.
Hence, in this case, $m(u)=u=u^0=u^1$.
\item
If $k(u)$ is colored $l$, then by definition $m(u)=u^0$ when $l=0$,
and $m(u)=u^1$ when $l>0$.
\end{itemize}
We conclude that $R(m(\bar u)) = R(\bar u^0) \in B_0$ when $l=0$,
and $R(m(\bar u)) = R(\bar u^1) \in B_1$ when $l>0$.
Hence, since $B_2 = B_0 \cup B_1$, we always have
$R(m(\bar u)) \in B_2$ as desired.

It remains to verify that $m$ maps the variables in $XY_hZ$ correctly.
If $u \in Y_h$, then $h(u) = k(u) \in Z'$
so $k(u)$ is colored white and $m(u)=u$ as desired.
If $u \in X - Y_h$, then by (a), $k(u)$ is colored 0 so $m(u)=u^0$ as desired.
Finally, if $u \in Z-Y_h$,
then by (b), $k(u)$ is colored $l>0$ so $m(u)=u^1$ as desired.
\qed

\bigskip
As a corollary, we obtain that the complexity of deciding logical
entailment for sifo CQs is not worse than that of deciding
containment for classical CQs:

\begin{corollary} \label{corentail}
Testing logical entailment of sifo CQs is NP-complete.
\end{corollary}
\begin{proof}
Membership in NP follows from Theorem~\ref{theorentail}; as a
witness for logical entailment we can use a homomorphism $h$
satisfying the first two conditions of the theorem, together with
a homomorphism $h_0$ from the query $(B,XY_hZ)$ to the query
$(B,XY_h) \Join (B,Y_hZ)$ witnessing the third condition of the
theorem.  NP-hardness follows because the problem has
containment of classical CQs as a special case, which is well
known to be NP-hard.
Indeed, logical entailment of a sifo $Q'$ by a sifo $Q$,
in the special case where the creation
functions of $Q$ and $Q'$ are nullary, amounts to classical containment
of $Q$ in $Q'$ when we ignore the function terms in the heads.
\end{proof}

\subsection{From oid-equivalence to logical entailment}

Let $Q$ and $Q'$ be sifo CQs of the
general forms (\ref{eqn:normal-form-one})
and (\ref{eqn:normal-form-two}) from
Subsection~\ref{sec:charoid}.  From our main Theorems
\ref{theoroid} and \ref{theorentail}, we can conclude the
following.

\begin{theorem}
\label{oid2log}
If $Q$ and $Q'$ are oid-equivalent, then $Q$ logically entails
$Q'$.
\end{theorem}
\begin{proof}
By Theorem~\ref{theoroid}, there exists
a permutation $\pi$ of $Z-X$ such that
$\fltg Q^\pi$ and $\fltg Q'$ are equivalent.
Hence there is a homomorphism $h : \fltg Q^\pi \to \fltg Q'$.
Clearly $h : B \to B'$.
We verify that $h$ satisfies the conditions of
Theorem~\ref{theorentail}, thus showing that $Q$ logically
entails $Q'$.
\begin{enumerate}
\item
Since $h$ maps the head of
$\fltg Q^\pi$ to the head of $\fltg Q'$, we have $h(\bar x)=\bar
x$ and $h(\pi(\bar z)) = \bar z$.
Since $\bar x'=\bar x$, we have
$h(\bar x) = \bar x'$ as desired.
\item
Since $h$ is the identity on $X$, we have
$h(X \cap Z) = X \cap Z \subseteq Z = Z'$ as desired.
\item
Since $h(\pi(\bar z))=\bar z$ and $\pi(Z)=Z$, we have $h(Z)=Z=Z'$.
Hence $Z \subseteq Y_h$.  But then the join dependency $X Y_h \Join Y_h Z$
becomes $X Y_h \Join Y_h$ which trivially holds.
\qedhere
\end{enumerate}
\end{proof}

\section{Discussion} \label{secdiscuss}

The results in this paper provide an understanding of the
notions of oid-equivalence and logical entailment for sifo CQs.
Sifo CQs, however, form a very simple subclass of oCQs.
Moreover, oCQs themselves are rather limited, for example, they
consist of a single rule and the rule can have only one atom in
the head.  Thus there are at least three natural directions for further
research: (i) allowing more than one function in the head; (ii)
allowing more than one atom in the head; (iii) allowing more than
one rule.

\paragraph{Containment}
Furthermore, in addition to oid \emph{equivalence} of oCQs, it
would be natural to also investigate a notion of
oid-\emph{containment}.  There are actually at least two
reasonable ways to define such a notion.  The situation is
similar to that in research on CQs with counting or bag semantics
\cite{cns_contagg,cohen_multiplicities}.  Most of the known results
are for equivalence only, with the extension to containment
typically an open problem.  Indeed, our characterization of
oid-equivalence for sifo CQs relies on equivalence of CQs with
bag semantics. An extension to oid-containment will likely need a
similar advance on containment of CQs with bag semantics.

\paragraph{Sifo CQs and ILOG}  In the introduction we mentioned
that sifo CQs, and oCQs in general, are a fragment of ILOG
without recursion \cite{hy_ilog}.  Sifo CQs belong to the
subclass of the class of recursion-free ILOG programs ``with
isolated oid creation'' \cite{hy_pods91}.  For this class,
oid-equivalence was already known to be decidable.  This was
shown by checking all finite instances up to some exponential
size.  Hence, our NP-completeness result for oid-equivalence of
sifo CQs does not follow from the previous work.  More generally,
the decidability of oid-equivalence for general recursion-free
ILOG programs, or already of oCQs for that matter, is a
long-standing open question.  Various interesting examples
showing the intricacies of this problem have already been given
by Hull and Yoshikawa \cite{hy_pods91}.

\paragraph{Sifo CQs and nested dependencies} In
Section~\ref{secnested} we also presented sifo CQs, now viewed as
schema mappings, as a very simple subclass of nested tgds.  The
implication problem for general nested tgds was shown to be
decidable by Kolaitis et al.~\cite{pichler_nestedep} in work done
independently from the present paper.  Nevertheless our
characterization of implication for sifo CQs, given by
Theorem~\ref{theorentail}, does not follow from the general
decision procedure for nested tgds.  Instead, the general
procedure, when applied to two sifo CQs, is strikingly similar to
our proof of necessity of our Theorem.  Using the notation from
that proof, the general procedure applied to test implication of
sifo CQ $Q'$ by sifo CQ $Q$ would amount to testing for the
existence of a homomorphism $h$ from $\{T(\bar x'^l,f'(\bar z'))
\mid l=0,\dots,n\}$ to $Q(I)$.  Since $Q(I) = \{T(\alpha(\bar
x),f(\alpha(\bar z))) \mid \alpha : B \to I\}$, this can be
implemented by guessing $h$ and $n+1$ matchings $\alpha_l : B \to
I$ such that $(h(\bar x'^l),f'(h(\bar z'))) = (\alpha_l(\bar
x),f(\alpha_l(\bar z)))$ for $l=0,\dots,n$.  In contrast, as
explained in Corollary~\ref{corentail}, our characterization
involves guessing just two homomorphisms.

\paragraph{Sifo CQs and plain SO-tgds}
As described in Section~\ref{seclog},
sifo CQs are a very simple subclass of plain SO-tgds.  
For plain {SO-tgds}, deciding logical equivalence is again an open problem.
Also, the notion of oid-equivalence, defined in this paper for oCQs,
can be readily extended to plain SO-tgds.  
We illustrate some difficulties involved in allowing
multiple functions in the head, which is indeed allowed in plain
SO-tgds.  First, consider the oid-equivalence problem.  For sifo
CQs we have shown in Section~4.4 of this paper that, as far as oid-equivalence
is concerned, only the \emph{counts} of generated oids per
tuple are important.  Now consider the following pair of oCQs:
\begin{align*}
Q & = T(x,f(y),g(x,z)) \gets R(x,y), R(x,z) \\
Q' & = T(x,f(y),g(x,y)) \gets R(x,y), R(x,z)
\end{align*}
Both queries create the same number of new $f$-oids and $g$-oids
per $x$-value, but now it also becomes important how these oids
are paired.  In $Q$ more pairs are generated for each $x$, and
the two queries are not oid-equivalent.  So, in the case of multiple
functions, also the interaction between the multiple terms needs
to be taken into account in some way.

A similar comment applies to the problem of
logical equivalence.
It is not immediately clear how the join dependency condition
of Theorem~\ref{theorentail} should be generalized in the
presence of multiple functions.  Consider,
for example, the following:
\begin{gather*}
Q  =
 T(x,f_1(z_1,y_1),f_2(z_2,y_2)) \gets
R(x,z_1,z_2), S(z_1,y_1), S(z_2,y_2)
\\
Q' =
 T(x,g_1(u),g_2(u)) \gets
R(x,u,x), R(x,x,u), S(u,v_1), S(x,v_2)
\end{gather*}
The $f_1$-part of $Q$ (ignoring the third component in the head)
logically entails the $g_1$-part of $Q'$, and likewise the
$f_2$-part of $Q$ (ignoring the second component in the head)
logically entails the $g_2$-part of $Q'$.  Globally, however, $Q$
does not logically entail $Q'$; this can be seen by the 
instances shown in Table~\ref{twofunfig}, which satisfy
$Q$ but not $Q'$.

\begin{table}
$$
\begin{array}[t]{|lll|}
\multicolumn{3}{c}{R} \\
\hline
2 & 1 & 2 \\
2 & 2 & 1 \\
3 & 1 & 3 \\
3 & 3 & 1 \\
\hline
\end{array}
\qquad
\begin{array}[t]{|ll|}
\multicolumn{2}{c}{S} \\
\hline
1 & 4 \\
2 & 5 \\
3 & 6 \\
\hline
\end{array}
\qquad
\begin{array}[t]{|lll|}
\multicolumn{3}{c}{T} \\
\hline
2 & 7 & 8 \\
2 & 8 & 7 \\
3 & 7 & 9 \\
3 & 9 & 7 \\
\hline
\end{array}
$$
\caption{Instances used to illustrate logical entailment in the
presence of multiple functions.}
\label{twofunfig}
\end{table}

A related interesting question then is whether
Theorem~\ref{oid2log}, that oid-equivalence implies logical
entailment, still holds for plain {SO-tgds}.  When we allow nested
function terms in the head (which goes beyond plain {SO-tgds})
the implication breaks down, as shown by the following example
\cite[Example 3.8]{nash_optimization}:
\begin{align*}
Q &= T(x,f(x),g(f(x))) \gets S(x) \\
Q' &= T(x,f(x),g(x)) \gets S(x)
\end{align*}
Here $Q$ and $Q'$ are oid-equivalent, and $Q$ logically entails
$Q'$, but $Q'$ does not logically entail $Q$.


\section*{Acknowledgment}

We thank the anonymous referees for their careful comments which helped
improve the presentation of the paper.

\end{document}